\documentclass[a4paper,12pt]{article}

\usepackage[mathscr]{eucal}
\usepackage{amsmath,amsfonts,amssymb,amsthm}
\usepackage{times}
\usepackage{hyperref}

\advance\textheight\topskip
\numberwithin{equation}{section} \makeatletter
\@addtoreset{equation}{section}

\newtheorem{prop}{Proposition}[section]

\renewcommand{\tilde}{\widetilde}
\renewcommand{\hat}{\widehat}

\newcommand{\bref}[1]{\textbf{\ref{#1}}}

\newcommand{\gh}[1]{\mathrm{gh}(#1)}

\newcommand{\dd}{\partial}
\renewcommand{\d}{\partial}

\newcommand{\inner}[2]{\langle #1{,}\,#2\rangle}
\newcommand{\binner}[2]{%
  {\langle}\kern-4.15pt{\langle}#1{,}\,#2{\rangle}\kern-4.15pt{\rangle}}
\newcommand{\commut}[2]{[#1{,}\,#2]}

\newcommand{\half}{\mathchoice{%
    \ffrac{1}{2}}{\frac{1}{2}}{\frac{1}{2}}{\frac{1}{2}}}

\newcommand{\ffrac}[2]{\raisebox{.5pt}%
  {\footnotesize$\displaystyle\frac{#1}{#2}$}\kern1pt}

\newcommand{\derham}{\boldsymbol{d}}

\newcommand{\dl}[1]{\mathchoice{\ffrac{\dd}{\dd #1}}{\frac{\dd}{\dd
      #1}}{\ffrac{\dd}{\dd #1}}{\ffrac{\dd}{\dd #1}}}

\newcommand{\st}[2]{{\overset{#1}{#2}}}

\newcommand{\ddl}[2]{\ffrac{\dd #1}{\dd #2}}

\newcommand{\vddl}[2]{{\ffrac{\delta #1}{\delta #2}}}

\newcommand{\manifold}[1]{\mathscr{#1}}

\newcommand{\manM}{\manifold{M}}

\newcommand{\manJ}{\manifold{J}}

\newcommand{\Liealg}{\mathfrak} 
\newcommand{\algg}{\Liealg{g}}

\def\cE{\mathcal{E}}
\def\cF{\mathcal{F}}

\def\cH{\mathcal{H}}

\def\cJ{\mathcal{J}}

\def\cL{\mathcal{L}}

\def\BG-Poincare{Barnich:2009jy}
\def\Fedosov-book{Fedosov:1996fu}

\newcommand{\cprime}{\/{\mathsurround=0pt$'$}}
\newcommand\Hor{\mathbf{H}}
\newcommand\dv{d_{\mathrm{v}}}
\renewcommand\dh{d_{\mathrm{h}}}
\renewcommand\DH{D_{\mathrm{h}}}
\newcommand\DV{D_{\mathrm{v}}}
\begin{document}
\begin{flushright}\small
FIAN-TD-2016-16
\end{flushright}

\begin{centering}

  \vspace{1cm}
   \thispagestyle{empty}
  \textbf{\Large{Presymplectic structures and intrinsic Lagrangians}}

  \vspace{1cm}

  {\large Maxim Grigoriev} 

 \vspace{1cm}
 
 \begin{minipage}{.9\textwidth} \it \begin{center}
    Tamm Theory Department, Lebedev Physics
    Institute,\\ Leninsky prospect 53,  119991 Moscow, Russia\\
    \vspace{0.5cm}   
Moscow Institute of Physics and Technology, Dolgoprudny,\\
141700 Moscow region, Russia  
    \end{center}
 \end{minipage}

\end{centering}

\vspace{1.5cm}
\begin{abstract}
%
%
%
%
%
   
    It is well-known that a Lagrangian induces a compatible presymplectic form on the equation manifold (stationary surface, understood as a submanifold of the respective jet-space). Given an equation manifold and a compatible presymplectic form therein, we define the first-order Lagrangian system which is formulated in terms of the intrinsic geometry of the equation manifold. It has a structure of a presymplectic AKSZ sigma model for which the equation manifold, equipped with the presymplectic form and the horizontal differential, serves as the target space. For a wide class of systems (but not all) we show that if the presymplectic structure originates from a given Lagrangian, the proposed first-order Lagrangian is equivalent to the initial one and hence the Lagrangian \textit{per se} can be entirely encoded in terms of the intrinsic geometry of its stationary surface. If the compatible presymplectic structure is generic, the proposed Lagrangian is only a partial one in the sense that its stationary surface contains the initial equation manifold but does not necessarily coincide with it.

%
%
%
%
%
%
%
%
%
%
%
%
%
%

\end{abstract}

\vspace{1cm}
\newpage

{
\tableofcontents
}

\section{Introduction}

Most of the theories of fundamental interactions are naturally Lagrangian theories. Although classical field dynamics can be described at the level of equations of motion, the Lagrangian (or its substitute) is inevitable at the quantum level.  Even classically, interactions are best described in the Lagrangian terms.  Moreover, the existence of a Lagrangian description is often considered as an important selection criteria that a candidate theory ought to satisfy. 

More practically, a typical question often met in applications is whether
the given equations of motion are Lagrangian (=variational) or not. This is known as 
the inverse problem of variational calculus 
(see e.g.~\cite{Henneaux:1982iw,Henneaux:1984ke,Anderson1988} for 
the introduction and original references). 
In its simplest version the question is whether a given system of partial differential 
equations (PDE) is the Euler-Lagrange equations derived from a local  Lagrangian.
Less trivial is the problem (known as the multiplier problem) whether the equation
is defined as a submanifold of a given jet-space. A more general question is whether a given PDE can be equivalently  reformulated as a Lagrangian one by performing a local invertible change of the variables
and/or by adding/eliminating so-called auxiliary fields. 

The difficulty in searching for a Lagrangian is that, on the one hand, the Lagrangian is defined on the jet-space (the space of all the dependent variables and their space-time derivatives, seen as independent coordinates), while, on the other hand, there is a huge ambiguity in realizing a given PDE in terms of one or another set of dependent variables so that apparently it is not clear which particular realization has a chance to be Lagrangian and which does not. A typical example is provided by the equations of motion of the massive spin-$2$ field whose Lagrangian formulation~\cite{Fierz:1939ix} requires introducing an auxiliary field.

A natural step is to try to formulate the problem in the invariant terms (=independent of
the particular embedding). In the invariant approach to PDE~\cite{Vinogradov1981} (for a modern exposition see e.g.~\cite{Bocharov:1999,Krasil'shchik:2010ij}), which is well-known by now, a PDE is defined as a manifold equipped with the Cartan distribution or, in more down-to-earth
terms, with a certain set of commuting vector fields. This manifold can be 
arrived at starting from a concrete realization of the PDE as a surface in the 
jet-space, singled out by the equations and their differential consequences. In 
so doing, the commuting vector fields determining the distribution are simply the 
total derivatives restricted to the surface. In contrast to the total derivatives,
the naive restriction of the Lagrangian to its stationary surface does not have 
much meaning and hence can not encode the Lagrangian formulation.

%

There is, however, a well-defined geometric structure that the Lagrangian does determine on 
the equation manifold. This is the so-called canonical presymplectic structure: a 
closed and conserved $(n-1,2)$-form (i.e. $n-1$-horizontal and $2$-vertical; $n$ 
stands for the space-time dimension) on the equation manifold. It was thoroughly 
discussed in the context of the covariant phase-space approach~\cite{Crnkovic:1986ex,Zuckerman:1989cx,Lee:1990nz,Barnich:1991tc}.
In the case of 1 space-time dimension, this presymplectic structure becomes a 
usual (pre)symplectic form which is nondegenerate if gauge symmetries are not 
present. This was shown to characterize variational equations in 1d and to 
encode the respective Lagrangian~\cite{Henneaux:1982iw}.

An attempt to generalize this to PDE was made by Khavkin~\cite{Khavkine2012}, 
based on the earlier important developments of~\cite{Bridges2009,Hydon:2005}. It was 
demonstrated that given a concrete realization of a PDE, any compatible presymplectic structure can be lifted to a Lagrangian whose stationary surface contains the equation manifold of the initial PDE.  However, this construction depends on an apparently arbitrary choice of the explicit realization.

An independent construction of a Lagrangian in terms of a presymplectic 
structure was proposed in~\cite{Alkalaev:2013hta} in the context of the 
super-geometrical description of gauge theories. In particular, it was 
demonstrated that given a manifold equipped with a presymplectic form compatible 
with a homological vector field, this data determines a natural Lagrangian in 
terms of the field, taking values in the manifold. In this way one can naturally 
reformulate nearly any Lagrangian gauge system, giving a geometrical setup for 
the so-called frame-like formulations well-known in the literature. This 
construction is deeply related to the BRST-BV formalism for gauge 
theories~\cite{Batalin:1983wj,Batalin:1984jr} (a useful pedagogical exposition 
can be found in e.g.~\cite{HT-book}) and can be seen as a presymplectic 
generalization~\cite{Alkalaev:2013hta} of the familiar AKSZ sigma model 
approach~\cite{Alexandrov:1995kv}.

In this work we demonstrate that the adapted version of the Lagrangian proposed in~\cite{Alkalaev:2013hta} can be defined for a generic equation manifold equipped with a compatible presymplectic structure. A remarkable feature of the construction is that the Lagrangian (called ``intrinsic'' henceforth) is built in terms of the intrinsic geometry of the equation manifold and does not refer to any particular realization of the equation. This is so because the dependent variables
for the intrinsic Lagrangian are the coordinates of the equation manifold itself. 
However, not all the coordinates give rise to genuine fields for the intrinsic Lagrangian because those on which the Lagrangian does not actually depend are interpreted as pure gauge ones and are to be eliminated, resulting in the formulation with finite number of dependent variables.
Note that formulating a given PDE in such a way that dependent variables are coordinates
on the stationary surface underlies the so-called unfolded formalism~\cite{Vasiliev:1980as,Lopatin:1987hz,Vasiliev:2005zu}, originally developed in the theory of higher spin fields.


It turns out that the intrinsic Lagrangian is in general only a partial one in the sense that 
its equations of motion are consequences of the original equations. However, we show that for a wide class of theories including, for instance, Yang-Mills theories and Einstein gravity, the intrinsic Lagrangian built out of  (a properly chosen representative for) the canonical presymplectic structure is equivalent to the initial one. It is important to note that 
not all physically interesting systems belong to this class. For instance massive spin-2 field (as well as massive higher spins) does not belong.


\section{Presymplectic form on the stationary surface}
\subsection{Jet-bundle and variational bicomplex}
Now we recall the basic notions of jet-bundle and variationsl calculus. Further details can be found
in e.g.~\cite{Krasil'shchik:2010ij,Bocharov:1999,Andersonbook}.

Without trying to be maximally general let us concentrate on a system of PDE with dependent variables $\phi^i$
and independent variables $x^a$, $a=1,\ldots,n$.  More geometrically, the starting point is the bundle $\cF$ over the space time (where $x^a$ are local coordinates) and whose fibres are cordinatized by $\phi^i$. For simplicity, we always work locally and avoid any global geometry subtleties.

The associated jet-bundle $\manJ=J^{\infty}(\cF)$ can be coordinatized by $x^a,\phi,\phi_a,\phi_{ab}, \ldots$. It is equipped with the total derivative 
\begin{equation}
 \d^T_a=\dl{x^a}+\phi^i_a\dl{\phi^i}+\phi^i_{ab}\dl{\phi^i_a}+\ldots
\end{equation} 
A Local form (function) $\alpha[\phi]$ on $\manJ$ is a differential form that can be represented as a pullback from $J^k(\cF)$ (finite-order jet-bundle) i.e. it depends on only a finite number of the coordinates. The exterior algebra $\Omega(\manJ)$ of local forms is equipped with the horizontal differential $\dh=dx^a\d^T_a$. The complementary differential $\dv\equiv d-\dh$ is called vertical. A generic local form can be decomposed into homogeneous ones of the form
\begin{equation}
 \alpha_{r,s}=\alpha_{a_1\ldots a_r}^{{I_1}\ldots I_s }[\phi] \dv \phi_{I_1}\ldots \dv \phi_{I_s} dx^{a_1}\ldots dx^{a_1} \,.
\end{equation} 
$\alpha_{r,s}$ is refereed to as  $(r,s)$-form ($s$-vertical and $r$-horizontal). Here ${I}$ stands for the multi-index of $\phi^i,\phi^i_a,\phi^i_{ab},\ldots$. This bigrading of $\Omega(\manJ)$ makes it into the bicomplex,
called variational bicomplex. The two differentials are $\dh$ and $\dv$. Note that 
\begin{equation}
\begin{gathered}
 \dh \dv+\dv \dh=0\,, \quad \dv^2=\dh^2=0\,, \quad \dh \dv \phi_{I}=dx^a \dv \phi_{aI}\,, \\ \dv \phi_{I}=(d-\dh)\phi_I=d\phi_I-dx^a \phi_{aI}\,,
\end{gathered}
\end{equation}
where $\phi_{aI}\equiv \d^T_a \phi_I$. Vertical forms vanish on total derivatives, i.e. $(\dv f)(\d^T_a)=0$.

To any $(n,l)$-form $\alpha$ one can associate its $\dh$ cohomology class $I\alpha$. More precisely,
such form is automatically $\dh$-closed (because $n$ is the space-time dimension) and hence is a representative of a $\dh$-cohomology class. It is convenient to chose a representative such that $I\alpha=\dv\phi^i\alpha_i$ for some $(n,l-1)$-forms $\alpha_i$.
$(n,l)$-forms considered modulo $\dh$-exact ones are called functional forms. The vertical
differential determines a so-called Euler operator $\delta^E=I\dv$ on functional forms. It is easy to check that $\delta^E\delta^E=0$, $\delta^E \dh=0$, and $I\dh=0$.

Among the vector fields on $\manJ$ an important subalgebra is formed by evolutionary vector fields. These are vertical (i.e. annihilating $x^a$) vector fields commuting with $\dh$ (or, equivalently with $\d_a^T$). Any evolutionary vector field is determined by its action on undifferentiated variables. A collection of local functions $f^i=f^i[\phi]$ gives rise to a unique evolutionary vector field $E_f$ such that $[\dh,E_f]=0$ and $E_f\phi^i=f^i$.

A system of partially differential equations (PDE) is a collection of local functions $E_\alpha[\phi]$ satisfying certain regularity assumptions. Together with all their total derivatives functions $E_\alpha$ determine a surface (called equation manifold or stationary surface) $\manM$ in $\manJ$.
More precisely, the surface is determined by 
\begin{equation}
 \d^T_{a_l}\ldots \d^T_{a_l} E_\alpha=0\,, \qquad l=0,1,2,\ldots
\end{equation} 
understood as the algebraic equations in $\manJ$. Because $\d^T_a$ preserves the ideal generated by the prolonged PDE, $\d^T_a$ is tangent to $\manM$ and hence restricts to $\manM$. It follows that $\dh$ restricts to $\Omega(\manM)$,  the algebra of local forms to $\manM$. Because $d$ descends to $\Omega(\manM)$ as well so does the vertical differential $\dv$.

In what follows we always assume that the equation does not constrain independent variables. More formally, just like jet-bundle itself $\manM$ is a bundle over space-time manifold. Two PDE are considered equivalent if the respective equation manifolds  $\manM$ and $\manM^\prime$ are isomorphic and the isomorphism sends $\dh$ on $\manM$ to $\dh$ on $\manM^\prime$. This justifies that a PDE can be defined
as a pair $(\manM,\dh)$.

A given PDE $(\manM,\dh)$ can be explicitly realized as an explicit system of 
PDE using one or another jet-bundle. There is however a somewhat distinguished 
realization, where the jet bundle is naturally determined by the equation 
manifold itself.  We discuss this realization in Section~\bref{sec:intr-embedd}.

%
%
%
%

\subsection{Lagrangian and the presymplectic structure}

The standard understanding of a variational PDE is as follows: equation $\manM \subset \cJ$,
where $\manM$ is understood as a submanifold of a given jet-bundle $\cJ$, is called variational
if there exist a local $(n,0)$-form $\cL=L[x,u](dx)^n$ such that the prolongation of
\begin{equation}
\label{EL}
 E_i=\ddl{\cL}{\phi^i}-\d^T_a\ddl{\cL}{\phi^i_a}+\d^T_a\d^T_b\ddl{\cL}{\phi^i_{ab}}-\ldots
\end{equation} 
determines $\manM$. The right-hand-side defines the Euler-Lagrange (EL) derivative of $\cL$.
Note that in the formulation where the equations are explicitly variational the number of equations coincides with the number of dependent variables. Here and below we employ the following useful notations:
\begin{equation}
(dx)^{n-k}_{a_1\ldots a_k}\equiv \frac{1}{(n-k)!}\epsilon_{a_1\ldots a_k c_1 \ldots c_{n-k}}
dx^{c_1}\ldots dx^{c_{n-k}}\,.
\end{equation}
Note the relation $dx^c (dx)^{n-k}_{a_1\ldots a_k}=(dx)^{n-k+1}_{[a_1\ldots a_{k-1}}\delta^c_{a_{k]}}$, where $[~]$ denotes the total antisymmetrization of the enclosed indices.

Given a variational PDE determined by the Lagrangian $\cL$ the naive restriction of $\cL$ to the equation manifold $\manM$ does not make much sense. However, the Lagrangian does determine an exact
$(n-1,2)$-form $\sigma$ on $\manM$ in a natural way. More precisely, one first defines an $(n-1,2)$-form $\hat\chi$ on the $\cJ$ by
\begin{equation}
\dv \cL=\dv \phi^i E_i - \dh\hat\chi\,,
\end{equation} 
where $E_i$ are the EL equations~\eqref{EL}. That $\hat\chi$ exists follows e.g. from the explicit structure of $E_i$. Then one takes presymplectic form $\hat\sigma$ to be $\hat\sigma=\dv \hat\chi$. Forms $\hat\chi$ and $\hat\sigma$ pulled back to $\manM$ are denoted by $\chi$ and $\sigma$ respectively.
It turns out (see e.g.~\cite{Khavkine2012}) that on the equation
\begin{equation}
 \dv\sigma=\dh\sigma=0\,.
\end{equation} 
Indeed, $\dv E_i$ pulled back to $\manM$ vanishes (because $dE_i$ clearly does so and $\dh E_i=0$) and hence $i^*_\manM(\dv \dh \hat\chi)=0$, where $i^*_\manM$  is the pull-back map. It follows $0=i^*_\manM(\dv \dh \hat\chi)=-i^*_\manM(d \dv \hat\chi)=-d \sigma$ which in turn implies $\dh\sigma=\dv\sigma=0$.

If instead of $L$ we started with $L^\prime=L+\dh \alpha$ this would result in $\hat\chi^\prime=\hat\chi-\dv\alpha$ and the same $\hat\sigma$ so that adding total derivative to $L$ doesn't affect $\sigma$. The form $\hat\chi$ is defined modulo $\dh$-exact. For $\hat\sigma$
this gives $\hat\sigma \sim \hat\sigma+\dv \dh \beta$ for some $n-2,1$-form $\beta$. Pulling this back to $\manM$ gives 
$\sigma \sim \sigma + \dv \dh (\beta|_\manM)$ (because $(\dh\beta)|_\manM=\dh(\beta|_\manM)$.

Let us explicitly compute $\hat\chi$ in the example of $\cL=L(\phi,\phi_a,\phi_{ab})(dx)^n$. One has
\begin{multline}
 \dv \phi^i E_i -\dv \cL=\\
 =\dv \phi^i(\ddl{L}{\phi}-\d^T_a\ddl{L}{\phi^i_a}+\d^T_a\d^T_b\ddl{L}{\phi^i_{ab}})(dx)^n-
 (\dv\phi^i \ddl{L}{\phi^i}+\dv \phi^i_a\ddl{L}{\phi^i_a}+ \dv \phi^i_{ab}\ddl{L}{\phi^i_{ab}})(dx)^n=\\
 =\dh\hat\chi=\dh\left(\left((\ddl{L}{\phi^i_a}- \d^T_b\ddl{L}{\phi^i_{ab}})\dv\phi^i+ \ddl{L}{\phi^i_{ab}}\dv\phi^i_b\right)(dx)^{n-1}_a\right)\,,
\end{multline} 
so that
\begin{equation}
 \hat\chi=\left((\ddl{L}{\phi^i_a}- \d^T_b\ddl{L}{\phi^i_{ab}})\dv\phi^i+ \ddl{L}{\phi^i_{ab}}\dv\phi^i_b\right)(dx)^{n-1}_a\,.
\end{equation} 
The generaliztaion to higher derivative Lagrangians is straitforward. 

In most of the application $L=L(x,\phi,\phi_a)$. In this case $\hat \chi $ and $\hat\sigma$ read explicitly as
\begin{equation}
\hat\chi=\dv \phi^i \ddl{L}{\phi^i_a} (dx)^{n-1}_a\,, \quad 
\hat\sigma=\dv \phi^i_a(\ddl{L}{\phi^i_a\d \phi^j}\dv\phi^j (dx)^{n-1}_a 
+
\ddl{L}{\phi^i_a\d \phi^j_b}\dv\phi^j_b \dv \phi^i)\,.
\end{equation}


Being closed, the presymplectic structure on $\manM$ should be exact $\sigma=dA$ for some form $A$. On $\manJ$ one has
$\hat\sigma=\dv \hat \chi=d\hat\chi-\dh\hat\chi=d\hat\chi-\dv\phi^i E_i + d \cL=d(\hat\chi+\cL)-\dv\phi^i E_i$.
By pulling back this equality to $\manM$ one gets
\begin{equation}
 \sigma=d(\chi+\cL|_\manM)\,.
\end{equation}

%
Following,~\cite{Khavkine2012}, for a generic equation $(\manM,\dh)$ we call presymplectic structure $\sigma$ compatible if $\dh\sigma=\dv\sigma=0$. 
The equation equipped with a compatible presymplectic structure is denoted by $(\manM,\dh,\sigma)$.

\subsection{Intrinsic embedding of a PDE}
\label{sec:intr-embedd}

Suppose we are given with an equation $(\manM,\dh)$ given in the intrinsic terms i.e. there is a manifold $\manM$ with coordinates $\psi^A,x^a$ equipped with $\dh=dx^a\d_a+dx^a Y_a$ where $Y^a=Y^B_a(\psi)\dl{\psi^B}$ such that $\dh^2=0$. Recall, that by assumption  $\manM$ is a bundle over the space of independent variables $x^a$. It is assumed that $(\manM,\dh)$ can be embedded into some jet-bundle but neither bundle nor the embedding is specified.

Starting from $(\manM,\dh)$ one can define an explicit realization of this equation.
Before giving an invariant definition let us first present a component one. To this end let us promote
all the coordinates $\psi^A$ on $\manM$ to the fields $\psi^A(x)$ of a new system with the same independnet variables $x^a$, and subject them to the following equations
\begin{equation}
\label{unfold-f}
\derham\psi^A(x)-(\dh \psi^A)(x)=0\,, \qquad \derham\equiv dx^a\dl{x^a}\,.
\end{equation} 
Here and below by $A(x)$ we denote a local horizontal form $A=A(\psi,x,dx)$ evaluated at $\psi^A=\psi^A(x)$. Note that space-time derivatives of $\psi^A$ enter only through $\derham \psi^A(x)$ because $\dh\psi^A=dx^a Y^A_a(\psi)$. The above equation is in fact equivalent to the starting point one. The idea to promote coordinates on the equation manifold of a given PDE to fields of a natural first-order reformulation of the same PDE underlies the so called unfolded formalism~\cite{Vasiliev:1980as,Lopatin:1987hz,Vasiliev:2005zu}. In particular, equations of the form~\eqref{unfold-f} are known as unfolded ones (note though that strictly speaking in contrast to the unfolded formulation in the present setting all the fields $\psi^A(x)$ are zero forms, even if the system has gauge symmetries).

To describe the above realization in a more invariant terms let us consider a new jet-space $J^\infty(\manM)$, namely the jet-bundle of the bundle $\manM$. In terms of coordinates, the new jet-space is coordinatized by $x^a,\psi^A_a,\psi^A_{ab},\ldots$. Let us stress that the number of dependent coordinates is infinite but as we are going to see only finite number of them are involved in the construction.

On the new jet space one defines a horizontal differential $\DH$ (we use different notation not to confuse with $\dh$) in a usual way
\begin{equation}
 \DH x^a=dx^a\,, \quad \DH \psi^A=dx^a \psi^A_a\,, \quad \DH \psi^A_a=dx^b \psi^A_{ab}\,, \quad \ldots\,.  
\end{equation} 
In the new jet-space consider an equation manifold determined by the prolonongation of
\begin{equation}
\label{unfold}
\DH\psi^A=\dh\psi^A \,.
\end{equation} 
It turns out that this equation manifold is isomorphic to the starting point one. Indeed, it is easy to see that $x^a,\psi^A$ can be chosen as coordinates on this manifold so that it can be identified with the original one while the above equations merely say that the horizontal differentials do coincide (a proof based on the use of Koszule-Tate differential
was given in~\cite{Barnich:2010sw}; the case of linear equations was alredy in~\cite{Barnich:2004cr})). 

\section{Intrinsic Lagrangian}

\subsection{Construction}
As we have just seen the equation $\manM$ can be embedded into the new jet-space $J^\infty(\manM)$ naturally build in terms of $\manM$ itself. It turns out that given a compatible presymplectic form $\sigma$ on $\manM$ there is a natural first-order Lagrangian defined on the new jet-space. It is called
the intrinsic Lagrangian henceforth.

The $(n-1,2)$-form $\sigma$ is closed and hence is exact (recall that we restrict ourselves to local analysis). It is also conserved $\dh \sigma=0$ and hence $\dv$-closed so that it is $\dv$-exact i.e. $\sigma=\dv\chi$ for some $\chi$. It follows it can be written as $\sigma=d(\chi+l)$ where $l$ is an $(n,0)$ form. Indeed, $d_h\sigma=0=\dh\dv\chi=-\dv\dh\chi$ and hence there exist $(n,0)$-form $l$ such that $\dh\chi=-\dv l=-dl$, giving $\sigma=d(\chi+l)$. As we have already seen in the case where $\sigma$ originates from the Lagrangian $\cL$ one can simply take $l=\cL|_{\manM}$ so that $\chi+\cL|_\manM$ is a pull-back of the generalized Poincare-Cartan form to the equation manifold.

Before giving an invariant definition of the intrinsic Lagrangian it is instructive to present a coordinate expression. Using coordinates $x^a,\psi^A$ on $\manM$ introduce vertical components of $\chi,\sigma$ according to $\chi=\dv \psi^A \chi_A $ and $\sigma=\half \dv \psi^A \dv\psi^B \sigma_{AB}$.  Note that $\chi_A$ and $\sigma_{AB}$ are horizontal forms. Promoting all the coordinates $\psi^A$ to fields $\psi^A(x)$ the action associated to the intrinsic Lagrangian reads as
\begin{equation}
\label{n-action}
 S^C=\int  \derham \psi^A(x) \chi_A(x) -\cH(x)\,, \qquad \derham \equiv dx^b\d_b, \quad  \cH=\dh \psi^A\chi_A -l\,.
\end{equation} 
$\cH$ is an $(n,0)$-form on $\manM$ called the covariant Hamiltonian. It satisfies
\begin{equation}
\label{dvH}
 \dv \cH=\dv(\dh\psi^B\chi_B)+\dh(\dv\psi^A \chi_A)
 =-\dv\psi^A(\dh \psi^B \sigma_{BA}+\derham \chi_A)\,.
\end{equation}
This can also be written as $\dv \cH=i_{\dh} \sigma+\derham \chi$, where 
$i_{\dh}$ denotes an operation which substitutes $\dv\psi^A$ with $\dh\psi^A$, i.e. symbolically $i_{\dh}=\dh \psi^A\dl{\dv\psi^A}$. In this form it is clear that~\eqref{n-action} belongs
to the class of Lagrangians put forward in~\cite{Alkalaev:2013hta}, as discussed in more details
in Section~\bref{sec:brst-like}. Note that a construction of $l$ and $\cH$ starting from $\sigma$ as well as an alternative Lagrangian construction was put forward in~\cite{Bridges2009}.

%
%

Using~\eqref{dvH} it is easy to write down explicitly the component form of the equations of motion following from $S^C$:
\begin{equation}
 \left(\derham \psi^B(x)-(d_H \psi^B)(x)\right)\sigma_{BA}(x)=0\,,
\end{equation} 
where $A(x)$ denotes evaluation of a horizontal form $A(\psi,x,dx)$ at $\psi^A=\psi^A(x)$. In this form it is clear that these are consequences of the equations~\eqref{unfold} and hence
of the initial equations of motion.

In a more invariant language the intrinsic Lagrangian is an $(n,0)$-form on the new jet-space $J^\infty(\manM)$ given by
\begin{equation}
 \cL^{C}=\Hor(\pi^*(\chi+l))\,,
\end{equation} 
where $\pi^*$ is the pullback associated to the projection $\pi$ of $J^\infty(\manM)$ to $\manM$ \footnote{in coordinate terms $\pi$  sends a point with coordinates $x^a,\psi^A,\psi^A_a,\psi^A_{ab},\ldots$ to $x^a,\psi^A,0,0,\ldots$.}
and $\Hor$ is the so-called horizontalization map. It sends a form on the new jet bundle to its completely horizontal component i.e. it does not affect coefficients while on the basis differentials it is defiend as 
\begin{equation}
\Hor(dx^a)=dx^a, \qquad \Hor(d\psi^A_{ab\ldots})=\DH\psi^A_{ab\ldots}\,.
\end{equation} 
For any local form $\alpha$ on $J^\infty(\manM)$ we have the following property:
\begin{equation}
 \Hor(d\alpha)=\Hor((\DH+\DV)\alpha)=\Hor(\DH\alpha)=\DH\Hor(\alpha)\,.
 \end{equation} 
In particular, if instead of $\chi+l$ we take $\chi+l+d\alpha$ this results in $(\cL^C)^\prime=\cL^C+\DH \Hor(\pi^*(\alpha))$, i.e. in adding a total derivative. This in turn implies that the equivalence class of the intrinsic Lagrangian modulo total derivatives is determined by $\sigma$ and does not depend
on the choice of the potential $\chi+l$.

Consider as an example a system whose Lagrangian $L=L(x,\phi,\phi_a)$ is independent of second and higher-order derivatives and is such that its equations of motion do not impose algebraic constraints on the dependent variables $\phi^i$. This means that $\phi^i|_\manM$ remain independent and can be taken as part of the coordinates on $\manM$, which we keep denoting by $\phi^i$. The form $\chi$ is then given explicitly by
\begin{equation}
 \chi=(d\phi^i-dx^b\d^T_b \phi^i) \left(\ddl{L}{\phi^i_a}\right)\Big|_\manM  (dx)^{n-1}_a\,.
\end{equation}
Because it is written in terms of De Rham differentials the component expression of its pullback to the new jet-bundle is unchanged. The decomposition into the new horizontal and vertical parts reads as
\begin{equation}
 \pi^*(\chi)=(\DV\phi^i+\DH\phi_i-\dh\phi^i)\left(\ddl{L}{\phi^i_a}\right)\Big|_\manM  (dx)^{n-1}_a\,,
\end{equation}
so that 
\begin{equation}
\label{intr-lag-exp}
\cL^C[\psi]= \Hor(\chi+\cL\big|_\manM)=
(\DH\phi_i-\dh\phi^i)\left(\ddl{L}{\phi^i_a}\right)\Big|_\manM (dx)^{n-1}_a+\cL\big|_\manM \,. 
\end{equation}

\subsection{Interpretation of the intrinsic Lagrangian}

Although the intrinsic Lagrangian is defined on the jet-bundle with infinite amount of dependent coordinates it actually depends on only the finite amount of them. It is natural to treat all the dependent variables on which $\cL^C$ does not depend, as  pure gauge ones and hence to disregard them (e.g. gauge-fix). More formally, suppose that after a local and invertible change of coordinates on $\manM$ the coordinates $\psi^A$ split into two groups $\varphi^i$ and $w^\mu$ such  that 
\begin{equation}
\vddl{{}^{EL}\cL^C}{w^\mu}=0\,.
\end{equation} 
This says that transformations $\delta w^\mu=\epsilon^\mu$, where $\epsilon^\mu$ are arbitrary functions of $x^a$, are gauge symmetries of the action $\int \cL^C$. Such gauge symmetries are known as  Stueckelberg or algebraic or, simply, shift gauge symmetries. These are to be gauge fixed by e.g. setting $w^\mu=0$ in $\cL^C$. Eliminating all such variables results in the action that does not anymore have Stueckelberg gauge symmetries.

In the case at hand this can be performed as follows: the variation of the intrinsic action under $\delta \psi^A(x)=\epsilon^A$, where $\epsilon^A$ are generic functions in $x^a$, is given by
\begin{equation}
 \int (\derham \psi^B(x)-(\dh\psi^B)(x)) \sigma_{BA}(x) \epsilon^A\,.
\end{equation} 
It follows that if $R^A_\alpha(\psi)$ are zero vectors of $\sigma_{AB}$, i.e. $\sigma_{AB} R^B_\alpha=0$
then $\delta \psi^A=R^A_\alpha \epsilon^\alpha$ is a symmetry for arbitrary $\epsilon^\alpha(x)$. Suppose that
we have found all linearly independent vertical vector fields $R_\alpha$ on $\manM$ such that $i_{R_\alpha}\sigma=0$.
It follows $i_{\commut{R_\alpha}{R_\beta}}\sigma=0$ and hence the distribution determined by $R_\alpha$ is integrable.
As our analysis is local we can find new coordinates $w^\mu,\varphi^i$ such that $R_\alpha=R_\alpha^\mu(\varphi,w)\dl{w^\mu}$ with $R_\alpha^\mu$ invertible and hence one can use the gauge symmetries to set $w^\mu=0$ (or any other convenient value), giving a natural set of fields for the intrinsic Lagrangian.

Note that although at first glance the above argument deals with infinite dimensional manifold no subtleties may arise. Indeed, all the objects entering $\cL^C$ originate from finite jets and hence may only involve finite amount of coordinates. This means that as a first step one can safely disregards infinite amount of variables which are not involved at all to reduce the problem to a finite-dimensional one.

\subsection{Dependence on the choice of presymplectic structure}

As we have already seen the Lagrangian doesn't uniquely determine the presymplectic structure. The ambiguity is described by the following equivalence transformation:
\begin{equation}
 \sigma \to \sigma +\dv\dh \alpha
\end{equation} 
where the $(n-2,1)$-form $\alpha$ is generic. Modulo $d$-exact terms this results in
\begin{equation}
 \chi+l\quad  \to \quad \chi+l+\dh\alpha\,.
\end{equation} 
The respective variation of the intrinsic Lagrangian is given by 
\begin{equation}
\Hor(\pi^*(\dh\alpha)) = \dh\Hor(\alpha)=\dh((\DH-\dh)\psi^A \,\alpha_A)
\end{equation} 
An equivalent (modulo $\DH$-exact terms) representation of the variation can be obtained starting from
$\chi+l~\to~ \chi+l-\dv\alpha$.
This gives
\begin{equation}
 -\Hor(\pi^*(\dv\alpha)) =-((\DH-\dh)\psi^A))((\DH-\dh)\psi^B))(\dv\alpha)_{AB}\,.
\end{equation} 
This is not always a total derivative. In other words the intrinsic Lagrangian 
does depend on the choice of $\sigma$ representative. Note however, that for a given 
Lagrangian system the ambiguity in $\sigma$ can be substantially reduced by 
requiring $\sigma$ to have minimal derivative order. As we are going to see in 
the next section, by using a minimal first-order formulation one can completely 
fix the ambiguity in $\sigma$.

\subsection{The statement}
Now we are going to compare the starting point Lagrangian $\cL[\phi]$ and the constructed above intrinsic Lagrangian $\cL^C[\psi]$. It is natural to consider two Lagrangians equivalent if they can be made identical (modulo total derivatives) by local invertible field redefinitions. Moreover, if by
such a redefinition the Lagrangian can be equivalently rewritten as $L^\prime=L[u]+L_a(v)$ where variables $v^\alpha$ enter only undifferentiated and $\vddl{{}^{EL}L}{v^\alpha}$ can be solved algebraically with respect to $v$ then $L^\prime$ is equivalent to $L[u]$. The equivalence of Lagrangians is stronger then the equivalence of the respective Euler--Lagrange equations. In particular, two equivalent Lagrangians determine equivalent equation manifolds and moreover the respective presymplectic structures are equivalent.

In what follows we assume that the initial Lagrangian $\cL$ does not have 
algebraic gauge symmetries. Moreover, we restrict ourselves to a class 
of natural Lagrangian systems defined as follows: a Lagrangian system is called natural if its
action is equivalent to the one of the form
\begin{equation}
\label{1st-order}
 S=\int \cL^{first}[\varphi]=\int d^n x (V^a_i(\varphi,x)\d_a \varphi^i-H(\varphi,x))\,,
\end{equation} 
and such that its equations of motion do not imply algebraic constraints on the undifferentiated fields $\varphi^i$. More precisely, the jet-space coordinate functions $\varphi^i$ pulled back to the equation manifold remain independent. Note that as a local form on the jet-space $\cL^{first}[\varphi]$ can be written as $\cL^{first}[\varphi]=\dh\phi^iV^a_i(dx)^{n-1}_a-H(dx)^n$.


Most of the theories of fundamental interactions (Einstein gravity, Yang-Mills, 
massless higher spin fields etc.). This can be easily seen by inspecting the 
well-known frame-like Lagrangians of gravity and Yang-Mills. In the case of 
massless higher-spins frame-like Lagrangians were proposed 
in~\cite{Vasiliev:1980as,Lopatin:1987hz,Skvortsov:2008sh,Skvortsov:2010nh}.
In fact there is a deep relation between frame-like Lagrangians and presymplectic 
structures observed in~\cite{Alkalaev:2013hta} but it becomes manifest 
only in the BRST extended version of the construction discussed briefly in 
Section~\bref{sec:grav}. In mathematical literature Lagrangian system of the form~\eqref{1st-order} are known as multisymplectic and were studied in~\cite{Hydon:2005,Bridges2009}.

We have to stress, however, that not all physically interesting systems are 
natural. For instance massive spin-2 (as well as massive higher spins) does not 
belong to this class.\footnote{In the standard approach~\cite{Fierz:1939ix} this can be 
traced to the zeroth-order differential consequences of the Euler-Lagrange equations.}

\begin{prop}
Let $\cL[\varphi]$ be a Lagrangian of a natural system. There exist a representative $\sigma$ of the equivalence class of presymplectic structures determined by $\cL[\varphi]$, such that the associated intrinsic Lagrangian $\cL^C[\psi]$ is equivalent to $\cL[\varphi]$.
\end{prop}
\begin{proof}
Equivalent Lagrangian formulations result in equivalent presymplectic structures on the equation manifold so that without loss of generality let us assume that we start with the first order Lagrangian~\eqref{1st-order}. The respective presymplectic structure reads as
\begin{equation}
 \hat\chi= \dv \varphi^i \,V^a_i(\varphi) \, (dx)^{n-1}_a\,.
\end{equation} 
If by slight abuse of notations $\varphi^i$ restricted to $\manM$ are also denoted by $\varphi^i$, then in the coordinate system on $\manM$ such that $x^a,\phi^i$ are part of the coordinates one has
\begin{equation}
\chi=\hat\chi|_\manM=\dv\varphi^i V^a_i(\varphi)\,(dx)^{n-1}_a \,.
\end{equation} 
Furthermore, the intrinsic Lagrangian \eqref{intr-lag-exp} takes the form
\begin{equation}
 L^C= \left(V^a_i\d_a \varphi^i- H(\varphi)\right)(dx)^{n}\,,
\end{equation} 
which explicitly coincides with the starting point first order Lagrangian~~\eqref{1st-order} provided one disregards all the dependent variables besides $\varphi^i$. Recall that according to our interpretation of the intrinsic Lagrangian all the variables of which it's independent, are to be gauged away.
\end{proof}

Let us note that in the above argument the representative $\sigma=\dv\chi$ of the presymplectic structure is quite distinguished. Indeed, the derivative order of $\hat\sigma$ is zero (only undifferentiated $\varphi^i$ enter). Any distinct  representative $\sigma^\prime=\sigma+\dv\dh \alpha$ necessarily involves derivatives of $\varphi$.

\subsection{Symmetries and conservation laws}

By definition a variational symmetry is a vertical evolutionary vector field $\hat V$ on $\manJ$ preserving the Lagrangian modulo a total differential, i.e. $\hat V x^a=0$, $\commut{\hat V}{\dh}=0$, and $\hat V\cL=\dh K$ for some $K$. It is clear that $\hat V$
is tangent to $\manM$ and hence determines an evolutionary vertical vector field $V=\hat V|_\manM$ on $\manM$. This preserves the presymplectic structure up to an equivalence. To see this let us rewrite explicitly the definition of $\chi$ using the Euler operator $\delta^E$:
\begin{equation}
 \dv\cL=\delta^E\cL-\dh\chi \,.
\end{equation} 
Applying the Lie derivative $L_{\hat V}$ to both sides and using $L_{\hat V}\delta^E\cL=\delta^E i_{\hat V} \delta \cL+I i_{\hat V} \delta^E\delta^E \cL= \delta^E i_{\hat V} \delta^E \cL=0$ which holds thanks to 
$i_{\hat V} \delta \cL = \dh (...)$, one finds: 
\begin{equation}
 \dv \dh K=\dh L_{\hat V}\hat\chi \,, \quad \Rightarrow \quad L_{\hat V}(\hat\chi+\dv K)=\dh \hat\alpha
\end{equation} 
for some $\alpha$. Applying $\dv$ gives
\begin{equation}
L_{\hat V}\hat\sigma=\dv\dh\hat\alpha\quad\Rightarrow \quad  L_{V}\sigma=\dv\dh\alpha\,,
\end{equation}
where the second equation is obtained by restricting to the equation manifold. This means that a variational symmetry preserves the equivalence class of $\sigma$.

Suppose that a vertical vector field $V$ on $\manM$ is a symmetry preserving the equivalence class of $\sigma$. I.e. 
\begin{equation}
 \commut{\dh}{V}=0\,, \qquad L_V \sigma=\dv\dh\alpha\,.
\end{equation} 
The presymplectic structure determines a map from the compatible symmetries to conservation laws. More precisely, let us define $(n-1,0)$-form $H_V$ by
\begin{equation}
 \dv H_V=i_V\sigma-\dh\alpha\,,
\end{equation} 
which is consistent because $\dv(i_V\sigma-\dh\alpha)=L_V\sigma-\dv\dh\alpha=0$. We have
\begin{equation}
 \dv\dh H_V=-\dh \dv H_V=-\dh (i_V\sigma-\dh\alpha)=i_V \dh\sigma=0\,,
\end{equation} 
where we made use of $\commut{\dh}{i_V}=0$ which holds thanks to $V$ being
vertical and evolutionary (note also that $i_V \dh f=0$ for any local function 
$f$). Hence $\dh H_V$ depends on $x^a,dx^a$ only. It follows one can assume 
$H_V$ satisfies $\dh H_V=0$. Indeed, as we work locally any 
$\psi^A$-independent $\dh$-closed $(n-1,0)$-form $\beta$ can be represented as 
$\beta=\dh\gamma$ for some $\psi^A$-independent $\gamma$. Such defined $H_V$ is 
an on-shell conserved horizontal $(n-1)$-form called the Hamiltonian of $V$. In 
the case where $\sigma$ originates from a genuine Lagrangian the above map is 
just the one of the Noether theorem and is one-to-one after modding out the 
gauge symmetries. For generic $\sigma$ the map is still defined but in general 
is not one to one.~\footnote{While preparing this work for publication we 
received Ref.~\cite{Sharapov:2016qne}, where the map from symmetries to conservation determined by a generic presymplectic structure  is also discussed
in the context of not necessarily Lagrangian system.  In this context it is also worth mentioning the dual structure (multidimensional generalization of Poisson bracket of the Hamiltonian formalism) that maps conservation laws to symmetries, see 
~\cite{Kersten:2003fb,Krasil'shchik:2010ij} in the context of integrable systems 
and \cite{Kaparulin:2010ab,Kaparulin:2011xy} in the context of gauge theories.}

Finally, consider the variation of the intrinsic Lagrangian under $\delta \psi^A=V\psi^A$. It is given by (modulo total derivatives)
\begin{equation}
 \delta \cL^C=(D_H \psi^B-\dh\psi^B)\sigma_{BA}V^A\,.
\end{equation} 
Taking into account that $V^A\sigma_{AB}=\d_B H_V+(\dh\alpha)_B$ one finds
\begin{equation}
 \delta \cL^C=D_H H_V-\dh H_V+(D_H \psi^B-\dh\psi^B)(\dh\alpha)_B\,.
\end{equation} 
The first term is a total derivative. The second one vanishes provided a proper choice of $H_V$. However, the third one is in general nonzero. For natural systems and properly chosen $\sigma$,  $V$ still determines a symmetry but its action of $\psi^A$ has to be modified. Note also that if $V$ strictly preserves $\sigma$ then $\delta \cL^C=D_H H_V$.

%
%
%
%
%

\subsection{Relation to parent action}

The intrinsic Lagrangian can be systematically derived from the so-called parent Lagrangian formulation~\cite{Grigoriev:2010ic,Grigoriev:2012xg}. To illustrate the relationship let us work in the simplified setting where $L=L(\phi,\phi_a,\phi_{ab})$ and no explicit $x^a$-dependence is allowed. 

Given a system with Lagrangian~$L=L(\phi,\phi_a,\phi_{ab})$ the respective parent action~\cite{Grigoriev:2010ic} reads as~\footnote{The first-order actions of this structure in 1 dimension (mechanics) are well-known, see e.g.~\cite{Gitman:1990qh}.}
\begin{equation}
\label{scalar-parent}
 S^P=\int dx^n \left(L(\phi,\phi_a,\phi_{ab})+ \pi^a(\d_a \phi - \phi_a)+\pi^{ac}(\d_a \phi_c - \phi_{ac}) +\ldots\right)\,.
\end{equation} 
where $\ldots$ denote further terms of the similar structure involving $\pi^{abc}$ etc. and all variables $\pi^{ab\ldots}$ are assumed totally symmetric. Its equations of motion read as
\begin{equation}
\label{parEOM}
\begin{gathered}
\ddl{L}{\phi}-\d_a\pi^a=0\,,\\
\pi^a-\ddl{L}{\phi_a}+\d_c\pi^{ca}=0\,,\qquad 
\pi^{ab}-\ddl{L}{\phi_{ab}}=0\,, \qquad \pi^{ab\ldots}=0\\
\phi_a=\d_a\phi\,, \qquad \phi_{ab}=\d_{(a}\phi_{b)}\,, \qquad \ldots\\
\end{gathered}
\end{equation}
It is easy to see that the Euler-Lagrange equations 
\begin{equation}
\label{EL-p}
 \ddl{L}{\phi}-\d^T_a\ddl{L}{\phi_a}+\d^T_c\d^T_a \ddl{L}{\phi_{ca}}=0\,.
\end{equation} 
determined by $L$ are consequences of~\eqref{parEOM}.

Considering the manifold $\bar\manM$ of independent variables $x^a$ and dependent variables of the parent system in place of $\manM$ one finds that the parent action can be written as
\begin{equation}
 S^P=\int (\derham\Psi^M \bar\chi_M -\bar\cH)\,,
\end{equation} 
where
\begin{equation}
 \bar\chi=(\pi^a \dv\phi+\pi^{ab}\dv\phi_b+\ldots)(dx)^{n-1}_a\,,
\end{equation} 
\begin{equation}
 \bar \cH\equiv
 \dh\Psi^A\chi_A-L\,\,(dx)^n=(\pi^a\phi_a+\pi^{ab}\phi_{ab}+\ldots-L(\phi,\phi_a,\phi_{ab}))(dx)^n\,.
\end{equation} 
and $\Psi^M$ denote all the dependent variables $\phi,\phi_a,\ldots$ and $\pi^a,\pi^{ab},\ldots$.  Note that only 
$\dh\pi^{a\ldots}$ do not actually enter the expressions while $\dh\phi_{\ldots}$ is defined as a usual horizontal differential on the jet-space of $\phi$, i.e. $\dh \phi_{\ldots}=dx^a\d^T_a\phi_{\ldots} =dx^a\phi_{a\ldots}$.

Consider the following submanifold  $\manM$ of  $\bar\manM$
\begin{equation}
\label{const}
\begin{gathered}
 \pi^a-\ddl{L}{\phi^a}+\d^T_c\ddl{L}{\phi_{ca}}=0\,,\qquad 
 \pi^{ab}-\ddl{L}{\phi_{ab}}=0\,, \quad
 \pi^{ab\ldots}=0\,,\\
 \d^T_{a_1}\d^T_{a_2}\ldots(EL)=0\,,
\end{gathered}
\end{equation} 
where $EL$ denotes EL equation~\eqref{EL-p}. The above constraints are (differential) consequences of the parent action equations of motion~\eqref{parEOM}. The submanifold they single out can be identified with the equation manifold $\manM$. Indeed, the last equation determines the equation manifold as a submanifold in jets-space (if one identifies coordinates $\phi,\phi_a,\ldots$ as those of the jet-space) while the first one puts $\pi^{\ldots}$ variables to the particular values.
%

Moreover, it is easy to check directly that the pullback of $\bar\chi$ and $\cH$ to $\manM$ explicitly coincides 
with the presymplectic potential $\chi$ and the covariant Hamiltonian $\cH$ determined by the Lagrangian
$L$ on its own equation manifold. This gives an alternative way to arrive at these structures.
Furthermore, under certain assumptions one can actually derive the intrinsic Lagrangian by eliminating
the auxiliary fields in the parent action. Let us also note that for gauge theories the parent action naturally extends~\cite{Grigoriev:2010ic,Grigoriev:2012xg} to the BV-BRST framework so that it can be used to derive a version of intrinsic Lagrangians whose gauge invariance is realized manifestly. In so doing the appropriate version of the  presymplectic form $\sigma$ originates from the odd symplectic structure of the parent BV formulation. It turns out that at least for usual gauge theories (gravity, YM theory, massless higher spins) such intrinsic Lagrangians  coincide with the familiar first-order frame like Lagrangians. 

%
%

\subsection{BRST-like description}
\label{sec:brst-like}

Till now we used the standard language of vertical and horizontal forms. It is instructive to reformulate the construction in the BRST-like language and to make contact with the presymplectic AKSZ models proposed in~\cite{Alkalaev:2013hta}. 

To this end we promote $dx^a$ to Grassmann odd ghost coordinates $\xi^a$. Both 
$x^a,\xi^a$ coordinates are then regarded as horizontal. In so doing a usual 
$(k,l)$-form becomes a vertical $l$-form which carries ghost degree $k$ while 
the horizontal differential $\dh$ becomes an odd nilpotent vector field 
$Q=\xi^a\d^T_{a}$ which acts on forms by the Lie derivative.  To simplify the 
exposition we assume that all the basic objects do not depend explicitly on the 
space-time coordinates~$x^a$.

In these terms the presymplectic structure is a vertical $2$-form of ghost
degree $n-1$ satisfying: 
\begin{equation}
\label{ldsigma}
 \dv\sigma=0\,, \qquad L_Q\sigma=0\,.
\end{equation} 
The definition of the covariant Hamiltonian $\cH$ takes the form
\begin{equation}
 i_Q\sigma=\dv \cH\,.
\end{equation} 
Using $\psi^A,x^a,\xi^a$ as coordinates on ghost-extended $\manM$ the above formula can be written as
$Q^A\sigma_{AB}=-\d_B\cH$. This can be solved in terms of the potential $\chi=\chi_A \dv\psi^A$ for $\sigma$ as 
\begin{equation}
\cH=i_Q\chi-l\,,  
\end{equation} 
where $l$ is a ghost-degree $n$ function. It is easy to check that this is the same $l$ as in~\eqref{n-action}.
In the case where $\sigma$ is determined by a Lagrangian $\cL$ one can take as $l$ the restriction of $\cL$
to the equation manifold.

Finally the expression for the intrinsic action takes the form
\begin{equation}
\int \derham \psi^A\chi_A -\cH
\end{equation} 
where $\psi^A$ is promoted to $\psi^A(x)$ while $\xi^a$ to $dx^a$.  If one also regards $x^a$ as another field, set to its background value, this action can be seen as that
of the presymplectic AKSZ sigma model~\cite{Alkalaev:2013hta} whose target space is $\manM$ extended by ghosts $\xi^a$. The only subtlety is that fields associated to coordinates $\xi^a,x^a$ are interpreted as background fields.~\footnote{These fields can be considered at the equal footing with others by considering the parameterized version of the same system. This also gives another (probably more 
fundamental) way to arrive at the BRST-like description. Parameterized systems in the presymplectic framework were discussed in~\cite{Alkalaev:2013hta}.}

\section{Examples}
\subsection{Pseudo 2nd order Lagrangian}
Let us consider a standard Klein--Gordon Lagrangian but written as $L=-\half\phi\eta^{ab}\phi_{ab}$ where $\eta^{ab}$ is the inverse metric. We have
(keep using $\phi,\phi_a$ to denote respective coordinates on the stationary surface)
\begin{equation}
 \chi=\left((\ddl{L}{\phi^a}-\d^T_c\ddl{L}{\phi_{ca}})\dv\phi+\ddl{L}{\phi_{ab}} \dv\phi_b \right)(dx)^{n-1}_a=\half(\phi^a \dv\phi-\phi\, \dv\phi^a)(dx)^{n-1}_a \,,
\end{equation} 
\begin{equation}
 \cH=\half(\phi^a\phi_a-\phi\eta^{ab}\phi_{ab}+\phi\eta^{ab}\phi_{ab})(dx)=\half\phi^a\phi_a(dx)^n\,.
\end{equation} 
The intrinsic action takes the form
\begin{equation}
 \half\int d^nx (\phi^a \d_a \phi- \phi \d_a \phi^a+\phi^a\phi_a) \,,
\end{equation} 
and indeed differs from a standard first-order action $\int(\phi^a \d_a \phi-\half \phi^a\phi_a)$
by a total derivative. Note that had we started with the usual Lagrangian $\half\phi_a\phi^a$
we would have arrived at the standard first order action.

\subsection{Polywave equation}

The simplest genuine higher derivative example is $L=\half\Box\phi \Box 
\phi=\half\phi_{aa} \phi_{bb}$ (here and below $\phi_{aa}= \eta^{ab}\phi_{ab}$ 
and as before we use $\phi,\phi_{a},\phi_{ab},\phi_{abc}$ as part of the 
coordinate system on the stationary surface). 
One has
\begin{equation}
 \chi=\left((\ddl{L}{\phi^a}-\d^T_d\ddl{L}{\phi_{da}})\dv\phi+\ddl{L}{\phi_{ab}} d\phi_b \right)(dx)^{n-1}_a=(-\phi_{acc}d\phi+\phi_{cc}  \dv\phi_a)(dx)^{n-1}_a 
\end{equation} 
and 
\begin{equation}
 \cH=(-\phi_{acc}\phi_a+\phi_{cc}\phi_{aa}-\half \phi_{cc}\phi_{aa})(dx)^n=(-\phi_{acc}\phi_a+\half \phi_{cc}\phi_{aa})(dx)^n\,.
\end{equation} 
The intrinsic action takes the form
\begin{equation}
\label{4ord}
 \int d^nx (-\phi_{acc}(\d_a \phi -\phi_a)+\phi_{cc}\d_a \phi_a-\half \phi_{aa} \phi_{cc})\,.
\end{equation} 
Note that the action depends on only the following variables $\phi,\phi_a,\phi_{cc},\phi_{acc}$ but not on the traceless component of $\phi_{ab}$
and $\phi_{abc}$. It is easy to check that this action is equivalent to the starting point one: indeed, varying with respect to $\phi_a$ and $\phi_{acc}$ gives $\phi_a=\d_a\phi$ and $\phi_{acc}=\d_a\phi_{cc}$ so that these equations can be algebraically solved for $\phi_a,\phi_{acc}$. Substituting the solution back to the action gives
\begin{equation}
 \int d^nx (\phi_{cc}\d_a \d_a\phi -\half \phi_{aa} \phi_{cc})\,.
\end{equation} 
Next, varying w.r.t. $\phi_{aa}$ gives $\phi_{aa}=\d_a\d_a\phi$. Substituting this
into the above action gives the starting point action.

The above example gives a nice illustration of how the intrinsic Lagrangian construction automatically selects a set of auxiliary fields required for the
minimal first-order formulation. More precisely the set of field consist of those coordinates the stationary surface on which the intrinsic Lagrangian actually depends (so that they survive the elimination of the pure gauge variables).

\subsection{YM theory}
The YM field is $A^a$ that takes values in a Lie algebra $\algg$ equipped with an invariant inner product $\inner{}{}$. We will use notation
$A^a_{b_l\ldots b_l}$ for $\d^T_{b_1}\ldots \d^T_{b_l} A^a$. The Lagrangian is given by (invariant summation over the repeated indices is assumed
)
\begin{equation}
 L=\frac{1}{4} \inner{F_{ab}}{ F_{ab}}(dx)^n\,, \qquad F_{ab}:=A^b_a-A^a_b+\commut{A^a}{A^b}\,.
\end{equation} 
Because $A^a_b$ are unconstrained by the equations of motion we
use $x^a,A^a,F_{ab},S_{ab}:=A^b_a+A^a_b$ restricted to the stationary surface
as part of the coordinate system therein.

The one form $\chi$ and the covariant Hamiltonian are given by
\begin{equation}
 \chi=\ddl{L}{A^b_a}\dv A^b(dx)^{n-1}_a=\inner{F_{ab}}{\dv A^b}(dx)^{n-1}_a \,,
\end{equation} 
 \begin{equation}
\cH=(\ddl{L}{A^b_a}A^b_a - \frac{1}{4} \inner{F_{ab}}{F_{ab}})(dx)=
 \half \inner{F_{ab}}{\half F_{ab}-\commut{A^a}{A^b}}                                                    \,. \end{equation} 
The intrinsic action takes the following form 
\begin{multline}
 \int \half \inner{F_{ab}}{\d_a A^b-\d_b A^a}-\half \inner{F_{ab}}{\half F_{ab}-\commut{A^a}{A^b}}
 =\\=
 \int \half \inner{F_{ab}}{\d_a A^b-\d_b A^a+\commut{A^a}{A^b}-\half  F_{ab}}\,,
\end{multline} 
and is clearly equivalent to the starting point action through the elimination of $F_{ab}$ by its own equations of motion. This is just the familiar first-order form of the YM action.

\section{Towards BRST extension: example of gravity}
\label{sec:grav}

Although all the above discussion applies to systems with gauge symmetries the 
gauge invariance was not explicitly taken into account. This can be 
systematically done using the BRST or more precisely (a generalization of the)  
Batalin--Vilkovisky formalism through the introduction of ghost variables and 
antifields. Here we only need a minimal set of structures.

Suppose that the PDE under consideration possesses gauge symmetries, i.e. a 
family of symmetries whose parameters are arbitrary functions of $x^a$. To 
describe gauge systems it is convenient to extend the set of dependent 
variables by ghosts $c^\alpha$ which are gauge parameters with the flipped 
Grassmann parity. We restrict ourselves to the case of irreducible gauge 
symmetries and hence
ghosts-for-ghosts are not present. It is also convenient to introduce a degree, 
called ghost degree, such  that $\gh{c^\alpha}=1$ while 
$\gh{x^a}=\gh{\phi^i}=0$. The jet-space is extended  to incorporate ghosts and 
their space-time derivatives $c^\alpha_{ab\ldots}$. The gauge transformations 
are encoded in the BRST differential $\hat\gamma$, which is an odd ghost degree 
$1$ vertical evolutionary vector field on the extended jet-space. $\hat\gamma$ 
is assumed to preserve the equation manifold $\manM$ and hence determines 
symmetry of the equation. That $\hat\gamma$ incorporates a compatible set of 
gauge symmetries, is encoded in the extra condition that $\gamma^2|_{\manM}=0$.

Suppose that the system is variational and let $\cL$ be the respective 
Lagrangian.
Gauge symmetries encoded in $\hat\gamma$ are said Lagrangian if
$\hat\gamma \cL=\dh j$ for some $(n-1,0)$-form $j$. Note that it can be natural 
to relax this condition to include symmetries  equivalent  to Lagrangian ones
(two symmetries are equivalent if they coincide on the stationary surface).
It is clear that  Lagrangian symmetries are automatically PDE symmetries. 

It is easy to read off gauge transformation of $\phi^i$ from $\hat\gamma$: 
namely if $\epsilon^\alpha$ are gauge parameters then 
\begin{equation}
 \delta_{\epsilon}\phi^i=\hat\gamma \phi^i|_{c^\alpha \to \epsilon^\alpha}\,.
\end{equation} 
That commutator of two gauge transformations is again a gauge transformation (on 
the stationary surface) is encoded in $\hat\gamma^2|_\manM=0$.

It is useful to consider a ghost-extended equation manifold $\cE$, which is $\manM$, extended by the ghost variables and their derivatives, and equipped with $\gamma=\hat\gamma|_{\cE}$. If we denote by $x^a,\psi^A,C^I$ the coordinates on $\cE$, where $C^I$ stand for all jet-space coordinate associated with the ghosts (i.e. $c^\alpha,c^\alpha_a,\ldots$) then $\gamma$ has the following form
\begin{equation}
 \gamma=C^I R_I^A(\psi)\dl{\psi^A}-\half C^I C^J U_{IJ}^K(\psi)\dl{C^K}\,.
\end{equation} 
It is clear that, thanks to $\gamma^2=0$, the vector fields $R_I$ determine an 
integrable distribution (called gauge distribution) on $\manM$ compatible with 
the Cartan distribution determined by $\dh$. To summarize, the ghost extended 
equation manifold is equipped with horizontal differential $\dh$ (which now also 
acts on ghosts) and the gauge differential $\gamma$ satisfying 
\begin{equation}
 \dh^2=0\,, \quad \dh\gamma+\gamma\dh=0\,, \quad \gamma^2=0\,.
\end{equation} 

In a direct analogy with the usual case discussed in 
Section~\bref{sec:intr-embedd}, given a ghost extended equation manifold 
$(\cE,\dh,\gamma)$ one can construct a natural realization of this gauge PDE in 
the intrinsic terms of $\cE$. More precisely, one promotes each coordinate 
on $\cE$ (besides $x^a,\xi^a$) to a field depending on $x^a$ which is a differential form whose degree is a ghost degree
of the coordinate. In our case $\psi^A$ give rise to the $0$-forms $\psi^A(x)$ while $C^I$ to 1-forms $A^I=dx^a A_a^I$. Introducing collective notation $\Psi^M$ for coordinates $\psi^A,C^I$ and $\Psi^M(x)$ for the associated fields
$\psi^A(x)$ and $A_a^I(x)dx^a$ the analog of the equations~\eqref{unfold-f} now reads as
\begin{equation}
 \derham \Psi^M(x)-((\dh -\gamma )\Psi^M)(x)=0\,.
\end{equation} 
Now $F(x)$, where $F=F(\psi,C,x,dx)$ is a horizontal form, denotes $F$ evaluated at  $\psi^A=\psi^A(x)$ and $C^I=A_a^I(x)dx^a$. The gauge symmetries of these equations are also determined by the differential $\dh-\gamma$ and can be written as
\begin{equation}
\delta\psi^A=-\epsilon^J(x)\ddl{(\gamma\psi^A)}{C^J}(x)\,,\qquad
\delta\ A^I=\derham \epsilon^I(x)
+\epsilon^J(x)\ddl{((\dh-\gamma)C^I)}{C^J}(x)\,.
\end{equation} 
More structural and uniform description is achieved in terms of the full-scale BRST-BV formalism and can be found in~\cite{Barnich:2010sw} (see also \cite{Grigoriev:2010ic,Grigoriev:2012xg}) where the above formulation was proved equivalent to the starting point one. 

An important property of the parent formulation is that contractible pairs for the differential $\gamma$ on $\cE$, which are by definition coordinates $w^a,v^a$ such that the equations $\gamma w^a=0$, $w^a=0$ are equivalent to $v^a=V^a(\text{remaining coordinates})$, correspond to the so-called generalized auxiliary fields which comprise usual auxiliary fields and pure gauge (Stueckelberg) variables. Addition/elimination of such variables leads to an equivalent realization of the gauge system. 

Let us concentrate on the case of diffeomorphism-invariant theories. Under rather general assumptions one can prove that by eliminating generalized auxiliary fields the system can be reformulated in such a way that $\dh$ disappears from the equations of motion and gauge symmetries. More precisely, the system becomes an AKSZ sigma-model whose target space is the ghost-extended equation manifold $\cE$ (but with coordinates $x^a$ eliminated) equipped with the differential $\gamma$ (see e.g.~\cite{Barnich:2010sw} and references therein for more details). In particular the equations of motion take the form
\begin{equation}
 \derham \Psi^M(x)+(\gamma \Psi^M)(x)=0\,.
\end{equation}
Formulation of a given gauge system in this form is also known as an unfolded one~\cite{Vasiliev:1980as,Lopatin:1987hz,Vasiliev:2005zu}. Strictly speaking, in the unfolded approach one typically deals with minimal (i.e. where maximal amount of the variables has been already eliminated) formulations of the above form. Moreover, in the general AKSZ formulation the variables of negative degree are present among $\Psi^M$, resulting in zeroth-order equations (constraints) among the equations of motion. Note that both approaches were developed independently from the quite different perspectives. Their relationship was described in~\cite{Barnich:2005ru}.

The AKSZ formulation is quite distinguished because it automatically contains the BRST formulation of the system. More precisely, let us promote a coordinate $\Psi^A$ to a collection of space-time forms of all degrees according to $\Psi^M=\st{0}{\Psi}{}^M(x)+\st{1}{\Psi}{}_a^M(x)dx^a+\ldots$ and set $\gh{\st{k}{\Psi}{}^M_{a_1\ldots a_k}}=\gh{\Psi^A}-k$ and the respective Grassmann parity.  It turns out that the introduced above ghost-degree $0$ component is precisely the $\gh{\Psi^A}$-form component while other components are identified as the ghost fields and the antifields needed for the BRST formulation of the system. The complete BRST differential is then determined as
\begin{equation}
 s\Psi^M(x,dx)=\derham \Psi^A(x,dx)+(\gamma \Psi^M)(x,dx)\,.
\end{equation} 
and is nilpotent by construction. Here $(\gamma\Psi^M)(x,dx)$ stands for local function $\gamma\Psi^M$ evaluated at $\Psi^M=\st{0}{\Psi}{}^M(x)+\st{1}{\Psi}{}_a^M(x)dx^a+\ldots$.

We now consider the example of Einstein gravity. In this case it is  known~\cite{Brandt:1996mh} that upon eliminating maximal amount of contractible pairs of $\gamma$ the reduced ghost-extended equation manifold $\tilde\cE$ can be coordinatized by
\begin{equation}
 e^a,~~\omega^{ab}\,, \qquad W^{ab}_{cd},~~W^{ab}_{cd;c_1},~~\ldots ~~W^{ab}_{cd;c_1\ldots c_l}, ~~\ldots
\end{equation} 
where the first group of variables have ghost-degree $1$ and the second $0$. Variables $e^a,\omega^{ab}$ originate from the diffeomorphism ghost and its antisymmetrized derivatives while $W$-variables can be related to the Weyl tensor and its algebraically-independent covariant derivatives. Note that $W^{ab}_{cd;\ldots}$ variables can be chosen totally
traceless. 

Upon the elimination of contractible pairs $\gamma$-differential on $\cE$
determines a reduced differential  $Q$ on $\tilde\cE$. Its explicit form is not known in general but it is easy to find how it acts on ghosts:~\footnote{\label{foot}More precisely, if one starts with the ghost-extended jet-space of gravity, then it is easy to eliminate contractible pairs for $\hat\gamma$. This results in the reduced jet-space equipped with the reduced differential $\hat Q$ whose structure is known explicitly (see e.g.~\cite{Brandt:1996mh}). The reduced differential on the equation manifold is then obtained by restricting $\hat Q$ to the equation manifold. Because equations impose no constraints on the ghosts and imply that Riemann tensor equals the Weyl tensor one immediately arrives at~\eqref{qgravred}. To find how $Q$ acts on covariant derivatives one needs to use the equations of motion to explicitly express the restriction of
$\hat Q$ in terms of the coordinates on the equation manifold. }
\begin{equation}
\label{qgravred}
 Q e^a=\omega^a{}_c\, e^c\,, \qquad Q\omega^{ab}=\omega^{a}{}_c \,\omega^{cb}+e^ce^d W^{ab}_{cd}\,, \qquad \ldots\,,
\end{equation}
The variables $e^a,\omega^{ab}$, and $W^{ab}_{cd;\ldots}$ provide a minimal formulation of the on-shell BRST complex for gravity and are known as generalized connection and tensor fields.
This concept is applicable to a general gauge theory and was put forward in~\cite{Brandt:1996mh,Brandt:1997iu}. Note that the supermanifold $\tilde\cE$ of this variables equipped with $Q$ data encodes all the information of the initial gauge theory. Indeed, as was shown in~\cite{Barnich:2010sw}, taking $\tilde\cE$ as a target space of the AKSZ sigma-model gives an equivalent formulation of the initial system so that the system
is reconstructed. This model is precisely the minimal unfolded formulation. 
Analogous considerations apply to nearly generic gauge theory.

Given $\tilde\cE$ equipped with $Q$ let us look for a compatible presymplectic structure, which in is this case is a 2-form $\sigma$ of ghost-degree $n-1$ satisfying (cf.~\eqref{ldsigma}):
\begin{equation}
 d\sigma=0 \,, \qquad \qquad  L_Q\sigma=0\,.
\end{equation} 
The respective presymplectic potential reads as~(this was proposed in~\cite{Alkalaev:2013hta})
\begin{equation}
\label{presympot}
\chi = d\omega^{ab} (e)^{n-2}_{ab}
\,,
\qquad \sigma=d\omega^{ab}de^c (e)^{n-3}_{abc}\,, 
\end{equation}
where
\begin{equation}
(e)^{n-k}_{a_1\ldots a_k}\equiv \frac{1}{(n-k)!}\epsilon_{a_1\ldots a_k c_1 \ldots c_{n-k}}
e^{c_1}\ldots e^{c_{n-k}}\,.
\end{equation}
It follows from the $o(n-1,1)$ invariance of $\epsilon_{a_1\ldots a_n}$ that $L_Q\chi=0$. The only subtle point in checking this is to observe that the Weyl tensor appearing in $Q\omega^{ab}$ does not contribute
because only its trace $W^{ca}_{cb}=0$ enters $L_Q\chi$. The covariant Hamiltonian is defined through $d\cH=-i_Q\sigma$ (we change sign for the sake of convenience) and is given by $\cH = i_Q\chi=\omega^{a}{}_c \,\omega^{cb}(e)^{n-2}_{ab}$. Promoting $e^a,\omega^{ab}$ to 1-form fields $e^a_\mu(x) dx^\mu,\omega^{ab}_\mu(x) dx^\mu$ the intrinsic action has the form of a presymplectic AKSZ model (see~\cite{Alkalaev:2013hta} for more details)
\begin{equation}
S^C=\int \derham \psi^A(x)\chi_A(x)-\cH=\int( \derham\omega^{ab}+\omega^{a}{}_c\omega^{cb})(e)^{n-2}_{ab}\,,
\end{equation}
and is just the usual gravity action in the frame-like formulation. Note that the action is explicitly independent of $W$-variables and hence these are to be disregarded. The above considerations easily extends to the case of nonvanishing cosmological constant.

The above construction is a slightly improved version of that from~\cite{Alkalaev:2013hta}. The important difference, however, is that the frame-like formulation is systematically constructed starting from the ghost-extended equation manifold. Nearly all the examples from~\cite{Alkalaev:2013hta} can easily be reformulated in the same way.

\section{Conclusions}

As concluding remarks let us discuss open problems and further perspectives.
First of all, a conceptual drawback of the proposed construction is the lack of an invariant characterization of a class of natural Lagrangian systems. We have only succeeded to characterize them implicitly as those systems whose Lagrangian can be brought to the specific first-order form by the local field redefinition and eliminating/adding auxiliary fields and/or pure gauge variables.

As we have seen for natural systems the Lagrangian formulation is encoded in the compatible
presymplectic structure on the equation manifold. The question is then how the Lagrangian formulation can be encoded in the intrinsic geometry of the equation manifold in the general case.

Given a compatible presymplectic structure which does not necessarily originate
from a Lagrangian (e.g. in the case where the Lagrangian is not known or does not exist) the intrinsic Lagrangian can still be used to perform (at least formally) a path-integral quantization of the system. In so doing the remaining equations of motion (those that do not follow from the intrinsic Lagrangian) are to be imposed as constraints. The idea to use a compatible presymplectic structure as a substitute of Lagrangian was also discussed recently in~\cite{Sharapov:2016qne}.


Finally, let us mention that the formalism developed in this work is closely related
to the de Donder--Weyl covariant Hamiltonian formalism (see,~\textit{e.g.},~\cite{Gotay:1997eg,Kanatchikov:1997wp,Kanatchikov:2000jz}). For instance,
in the simplest cases the covariant Hamiltonian $\cH$ coincides with the one of the de--Donder Weyl approach (see the respective discussion in~\cite{Alkalaev:2013hta}). In spite of this similarity, the detailed relationship is not known in the general case. 

\section*{Acknowledgments}
\addcontentsline{toc}{section}{Acknowledgments}

A substantial part of this work has originated from discussions with A.~Verbovetsky whom I wish to thank for his collaboration. In particular, he proposed the expression for the intrinsic Lagrangian as the horizontal component of the potential for the presymplectic form. I am also grateful to K.~Alkalaev and G.~Barnich for their useful exchanges. This work is supported by Russian Science Foundation grant 14-42-00047.

\addtolength{\baselineskip}{-3pt}
\addtolength{\parskip}{-3pt}

\providecommand{\href}[2]{#2}\begingroup\raggedright\endgroup

 \end{document}